\documentclass[journal,onecolumn]{IEEEtran}
\usepackage{amssymb, amsmath, amsthm, bm, tikz, cite}

\newtheorem{theorem}{Theorem}
\newtheorem{corollary}{Corollary}

\renewcommand{\d}{~\text{d}}
\renewcommand{\ae}{\text{a\hspace{0.2pt}e}}
\newcommand{\rect}{\text{rect}}
\newcommand{\nae}{\not\stackrel{\text{a\hspace{0.2pt}e}}{=}}

\title{On the Capacity of Waveform Channels Under Square-Law Detection of Time-Limited Signals}
\author{Amir Tasbihi~\IEEEmembership{Graduate Student
Member, IEEE} and Frank R.
Kschischang~\IEEEmembership{Fellow, IEEE}\thanks{The authors
are with the Edward S. Rogers Sr.\ Dept.\ of Electrical and
Computer Engineering, University of Toronto, Toronto, ON M5S
3G4, Canada.  Email:
\texttt{\{tasbihi,frank\}@ece.utoronto.ca.}. Submitted to {\it IEEE Trans. Inf. Theory}, January 8, 2020.}}

\begin{document}
\maketitle
\begin{abstract}
	\label{sec:abstract}
	Capacity bounds for waveform channels under square-law
	detection of time-limited complex-valued signals are
	derived. The upper bound is the capacity of the channel
	under (complex-valued) coherent detection. 
	The lower bound is one bit less, per dimension,
        than the upper bound.
\end{abstract}

\begin{section}{Introduction}
	\label{sec:introduction}
	\IEEEPARstart{S}{quare-law} detection (SLD) decides
	based on the squared
	magnitude of the received
	complex-valued waveform, contrasting with  
	coherent detection, in which the decision
	is based upon the received
	complex-valued waveform.
	The former 
	appears in many fields, {\it e.g.,} short-haul 
	fiber-optic communication systems~\cite{hecht},
	astronomical imaging~\cite{ligo},
	X-ray crystallography~\cite{xray},
	etc.

	As the measurement in SLD 
	depends on the magnitude of the received complex-valued signal, 
	it is often thought that 
	half of the degrees of freedom for data transmission are lost,
	when using this type of detection.
	Specifically, given a non-negative waveform $s(t)$, there are
	many complex-valued waveforms $y(t)$
	such that $|y(t)|^2=s(t)$. 
	Under some conditions on $y(t)$, there are algorithms
	that retrieve the phase of $y(t)$ from $s(t)$. This issue
	is well studied in the literature
	on {\it phase retrieval}, {\it e.g.,}
	see \cite{phase1,phase2,phase3,phase4,phase5,phase6}.

	Although studying the number of bandlimited $T$--periodic
	complex functions with the same magnitude goes back more than half a
	century~\cite{walther}, its direct consequence in finding a capacity
	lower-bound for SLD of bandlimited
	signals is recent~\cite{mecozzi}. Specifically, it was 
	shown in~\cite{mecozzi} that by using SLD,
	at most 1 bit per degree of freedom is lost,
	in comparison with 
	complex-valued coherent detection,
	which suggests that noncoherent detection
	may remain a viable approach for emerging applications
	in short-haul fiber-optic communication systems.

	In practice, signals are time-limited and it is the purpose of this
	paper to find the relative capacity of channels under SLD
	of time-limited signals in comparison with
	complex-valued coherent detection.
	
	We adopt
	a similar method as in~\cite{mecozzi}, except that we 
	use a weaker condition for distinguishability of two signals.
	Two functions $y_1$ and $y_2\in\mathbb{C}^\mathbb{R}$ 
	are said to be {\it equal almost everywhere} (a.e.),
	written $y_1\overset{\ae}{=}y_2$, if
	\[\int_{\mathbb{R}}\left|y_1(t)-y_2(t)\right|^2\d t=0;\]
	when $y_1$ and $y_2$ are not equal a.e., we write 
	$y_1\nae y_2$.
	It can be shown that almost-everywhere equality is 
	an equivalence relation.
	Two functions $y_1$ and $y_2\in\mathbb{C}^\mathbb{R}$ 
	are said to be
	{\it equal up to a phase offset}, written 
	$y_1\overset{\phi}{\sim} y_2$, if there is a $\phi\in[-\pi,\pi)$ such
	that $y_1\overset{\ae}{=}\exp(i\phi)y_2$.
	When $y_1$ and $y_2$ are not equal up to a phase offset,
	then we write $y_1\overset{\phi}{\nsim}y_2$.
	Note that $y_1\overset{\ae}{=}y_2$ implies 
	$y_1\overset{\phi}{\sim}y_2$, but not conversely.
	The relation $\overset{\phi}{\sim}$ is obviously
	reflexive and symmetric, and transitivity follows from
	the Cauchy-Schwarz inequality; thus $\overset{\phi}{\sim}$
	is an equivalence relation.
	
	The authors of~\cite{mecozzi} assume that a coherent
	detector can distinguish $y_1$ from $y_2$ if and only
	if $y_1\overset{\phi}{\nsim}y_2$. Here, we assume the 
	relaxed condition that $y_1$ and $y_2$ are distinguishable by
	a coherent detector if and only if $y_1\nae y_2$.

	The rest of the paper is organized as follows. The problem setup is introduced
	in Sec.~\ref{sec:setup}.
	In Sec.~\ref{sec:blaschke}, some complex analysis tools
	are introduced, to be used 
	in Sec.~\ref{sec:capacity_bounds} in finding capacity bounds
	of channels under SLD relative to 
	coherent detection. In parallel to the 1-bit capacity gap
	for the bandlimited signals, which is established by~\cite{mecozzi},
	we derive the same gap for time-limited signals in Sec.~\ref{sec:capacity_bounds}.  
	In Sec.~\ref{sec:discussions}, the paper concludes with a brief
	discussion of how these results can be generalized.

	Through this paper, $\mathbb{N}, \mathbb{R}$, $\mathbb{R}^+$ and $\mathbb{C}$ denote the 
	set of non-negative integers, real, non-negative real, and complex numbers, respectively.
	The reciprocal conjugate of $\alpha\in\mathbb{C}$ is denoted by $\alpha^{-\ast}$; hence $\alpha^{-\ast}=(\alpha^\ast)^{-1}$.
	The polynomial ring over $\mathbb{C}$ is denoted by $\mathbb{C}[z]$, and for an integer $n$,
	$\mathbb{C}^{\leq n}[z]$ denotes the set of polynomials in $\mathbb{C}[z]$
	of degree at most $n$.
	The unit circle, {\it i.e.},
	$\{z\in\mathbb{C}~:~|z|=1\}$,
	is denoted by $\mathbb{T}$,
	$\mathbb{D}$ denotes the open unit disk, {\it i.e.},
		$\mathbb{D}\triangleq\{z\in\mathbb{C}~:~|z|<1\}$,
	$\overline{\mathbb{D}}$ denotes the closure of $\mathbb{D}$,
		{\it i.e.,} $\overline{\mathbb{D}}\triangleq \mathbb{D}\cup\mathbb{T}$, and 
	$\mathcal{A}(\mathbb{D})$ denotes the set of analytic functions on
		$\mathbb{D}$ that extend continuously to $\overline{\mathbb{D}}$.
		Finally, the rectangular function is defined as
	\[\rect(t)=\left\{\begin{array}{lc}
		1, & 0\leq t<1;\\
		0, & \text{otherwise.}
	\end{array}\right.\]
\end{section}

\begin{section}{Problem Setup}
	\label{sec:setup}
	A complex-valued signal, $x(t)$, whose support
	is a subset of $[0,1)$ is transmitted over a channel, 
	and a complex-valued signal, $y(t)$, whose support 
	is a subset of $[0,1)$, is received. Note that the supports
	of $x$ and $y$ might be different; for example,  
	channel dispersion might broaden the support of $y$
	in comparison with $x$, or the channel might compress the support. 
	The choice of support interval does not affect the generality of the 
	results of the paper, as is explained in Sec.~\ref{sec:discussions}.
	
	We
	assume that $x$ and $y\in \mathcal{L}^4[0,1)$, {\it i.e.,}
	$\int_{0}^{1}|x(t)|^4\d t<\infty$, and similarly for $y$. 
	The reason for this choice of function space will be clarified
	later in this section.

	Two receivers are compared.
	The {\it coherent receiver} decides on the transmitted waveform by observing
	$y$, while the {\it SLD receiver} decides on the transmitted waveform
	by observing
	$s(t)\triangleq |y(t)|^2$. Since $y\in\mathcal{L}^4[0,1)$, the waveform
	$s$ belongs to $\mathcal{L}^2[0,1)$, {\it i.e.,} $\int_{0}^{1}|s(t)|^2\d t<\infty$.
	The relationships among $x, y$, and $s$
	are shown in Fig~\ref{fig:model}.

\begin{figure}
\centering
	\includegraphics[scale=1]{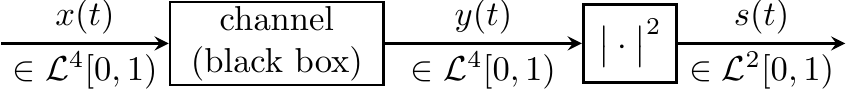}
	\caption{The system model}
	\label{fig:model}
\end{figure}

	As $y(t)$ is time-limited to $[0,1)$, we may assume that
	$y(t)=y_p(t)\rect(t),$
	where
	\[y_p(t)=\sum_{k=-\infty}^{\infty}y(t-k),\]
	is the periodic extension of $y(t)$ with period $1$.  

	According to Carleson's theorem~\cite{carleson},
	if a signal is in $\mathcal{L}^2[0,1)$ then its
	periodic extension is equal a.e. to its Fourier series.
	Note that 
	$\mathcal{L}^4[0,1)\subset\mathcal{L}^2[0,1)$~\cite{lp};
	as a result,
	\begin{equation*}
		y_p(t)\overset{\ae}{=}\sum_{k=-\infty}^{\infty}b_ke^{i2\pi kt},
		\label{eq:convergence}
	\end{equation*}
	where 
	\[b_k= \int_{0}^{1}y(t)e^{-i2\pi kt}.\]
	We can write $y_p(t)$ as $y_p(t)\overset{\ae}{=}\lim_{m\rightarrow\infty}y_{p,m}(t)$,
	in which 
	\[y_{p,m}(t)\triangleq\sum_{k=-m}^{m}b_ke^{i2\pi kt},\]
	is a truncated Fourier series.
	Writing $y(t;m)\triangleq y_{p,m}(t)\rect(t)$, 
	we then have 
	$y(t)\overset{\ae}{=}\lim_{m\rightarrow\infty}y(t;m).$
	Note that there is a one-to-one correspondence 
	between $y(t;m)$ and $\bm{y}^{2m+1}\triangleq(b_{-m},\ldots, b_m)\in\mathbb{C}^{2m+1}$.
	Similarly, let 
	\[x(t;m)\triangleq\left(\sum_{k=-m}^{m}a_ke^{i2\pi kt}\right)
	\rect(t),\quad a_k\in\mathbb{C},\]
	so that
	$x(t)\overset{\ae}{=}\lim_{m\rightarrow\infty}x(t;m).$ Then, we 
	can determine $x(t;m)$ uniquely from $\bm{x}^{2m+1}=(a_{-m},\ldots, a_m)\in\mathbb{C}^{2m+1}$.

	Square-law detection of $y(t;m)$ produces 
	$s(t;m)\triangleq|y(t;m)|^2$, which can be written as
	\begin{align*}
		s(t;m)&=\left|\left(\sum_{k=-m}^{m}b_ke^{i2\pi kt}\right)\rect(t)\right|^2\\
		&=\left(\sum_{k=-2m}^{2m}\sum_{\ell=\max(k-m,-m)}^{\min(k+m,m)}b_\ell b_{\ell-k}^\ast e^{i2\pi kt}\right)\rect(t)\\
		&=\left(\sum_{k=-2m}^{2m}c_ke^{i2\pi kt}\right)\rect(t),
	\end{align*}
	where
	\[c_k=\sum_{\ell=\max(k-m,-m)}^{\min(k+m,m)} b_\ell b_{\ell-k}^\ast.\]
	and we have used from this property that $|\rect(t)|^2=\rect(t)$. 
	Since $s(t;m)$ is a real-valued signal, we have $c_k=c_{-k}^\ast$.
	Similar to $\bm{x}^{2m+1}$ and $\bm{y}^{2m+1}$, there is a 
	one-to-one correspondence between $s(t;m)$ and
	$\bm{s}^{2m+1}=(c_0,\ldots, c_{2m})\in\mathbb{C}^{2m+1}$. 

	As $y\in \mathcal{L}^4[0,1)$, it implies that $s\in \mathcal{L}^2[0,1)$,
	which implies that
	$s(t)\overset{\ae}{=}\lim_{m\rightarrow\infty}s(t;m).$
	This is the reason that $y$ is considered to be
	in $\mathcal{L}^4[0,1)$, as in that case, $s$
	belongs to $\mathcal{L}^2[0,1)$ and Carleson's 
	theorem guarantees equality a.e. to $s(t;m)$, in the limit as
	$m\rightarrow\infty$.
		
	In summary, the system shown in Fig.~\ref{fig:model} behaves 
	like the system shown in Fig.~\ref{fig:periodic}, in the limit 
	as $m\rightarrow\infty$.\\
	\begin{figure}
	\centering
		\includegraphics[scale=1]{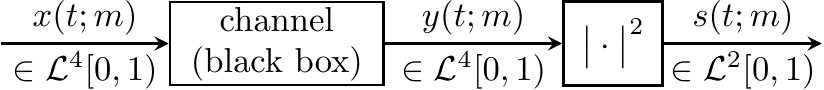}
		\caption{The actual system, shown in Fig.~\ref{fig:model},
		is equivalent to this system, when $m\rightarrow\infty$.} 
		\label{fig:periodic}
	\end{figure}

	The average mutual information between the time-limited functions
	$x(t;m)$ and $y(t;m)$ is defined as
	\[I_m(x(t;m);y(t;m))\triangleq\frac{I(\bm{x}^{2m+1};\bm{y}^{2m+1})}{2m+1},\]
	and similarly, for $x(t;m)$ and $s(t;m)$ as
	\begin{equation}
		I_m(x(t;m);s(t;m))\triangleq\frac{I(\bm{x}^{2m+1};\bm{s}^{2m+1})}{2m+1},
		\label{eq:infodef}
	\end{equation}
where $I(\cdot;\cdot)$ denotes the mutual information function. 
If the average mutual information per degree of freedom between
$x(t)$ and $y(t)$ exists, then it is given by~\cite[ch.~8]{gallager} 
\[I(x(t);y(t))=\lim_{m\rightarrow\infty}I_m(x(t;m);y(t;m)).\]
Similarly, if the average mutual information between $x(t)$ and 
$s(t)$ exists, then it can be written as
\[I(x(t);s(t))=\lim_{m\rightarrow\infty}I_m(x(t;m);s(t;m)).\]
Note that $I(x(t);y(t))$ and $I(x(t);s(t))$ are normalized to the 
number of used dimensions; as a result, they are similar to
the spectral efficiency under coherent detection and SLD,
respectively.

In this paper, we establish bounds for $I(x(t);s(t))$, in terms
of $I(x(t);y(t))$. To this aim, we establish bounds for
$I_m(x(t;m);s(t;m))$, in terms of $I_m(x(t;m);y(t;m))$ and 
we then let $m\rightarrow\infty$.\\ 

\end{section}

\begin{section}{On Blaschke Products}
	\label{sec:blaschke}
	To find a capacity lower-bound for the system shown in
	Fig.~\ref{fig:periodic}, we require some tools from complex analysis. 

		For $\alpha\in\mathbb{D}$, the {\it Blaschke factor},
		$B_\alpha:\overline{\mathbb{D}}\rightarrow\overline{\mathbb{D}}$, is defined as 
		\[B_\alpha(z)\triangleq\frac{\alpha-z}{1-\alpha^\ast z}.\]  
		Given a sequence $\alpha_1, \alpha_2,\ldots\in\mathbb{D}$, such that
		\[\sum_{k}(1-|\alpha_k|)<\infty,\]
		and $\tau\in\mathbb{T}$, the {\it Blaschke product}, $B(z)$, is
		defined as
		\[B(z)=\tau\prod_{k}B_{\alpha_k}(z).\]
		Furthermore, if $k$ is bounded above, then $B(z)$
		is called a {\it finite Blaschke product}. In general,
		a finite Blaschke product takes the form 
		\begin{equation}
			B(z)=\tau z^{n_0}\prod_{k=1}^{p}B_{\gamma_k}^{n_k}(z),
			\label{eq:bl1}
		\end{equation}
		for some finite $p\in\mathbb{N}$, some $\tau\in\mathbb{T}$,
		some distinct $\gamma_1,\ldots,\gamma_p\in\mathbb{D}\backslash\{0\}$,
		and some $n_0,\ldots,n_p\in\mathbb{N}$. The $z^{n_0}$ factor 
		in (\ref{eq:bl1}) corresponds to the Blaschke factor $B_0(z)=-z$.

		For any $\alpha\in\mathbb{D}$ and any $z\in\mathbb{T}$ we have $|B_\alpha(z)|^2=1$;
		as a result $B_\alpha(z)\in\mathbb{T}$. Consequently, 
		any Blaschke product maps the unit circle to itself.
		
		For a polynomial $f\in\mathbb{C}[z]$
		let $Z_f=\{\alpha\in\mathbb{C}~:~f(\alpha)=0\}$
		be the zero set of $f$.
		Let 
		\[Z_f'=\left\{\alpha\in\mathbb{D}\backslash\{0\}~:~f(\alpha)f\left(\alpha^{-\ast}\right)=0\right\},\]
		be the set of points, $\alpha\in\mathbb{D}\backslash\{0\}$,
		such that either $f(\alpha)=0$ or $f(\alpha^{-\ast})=0$.
		Finally, let $Z_f''=Z_f\cap\mathbb{T}$
		be the set of zeros of $f$ that are on the unit circle.
		We extend the usual notion of root-multiplicity to
		the entire complex plane as follows.
		For an arbitrary
		$\alpha\in\mathbb{C}$,
		let $d_f(\alpha)$ be the multiplicity of $\alpha$ as a root of $f$; 
		if $\alpha\notin Z_f$, then let $d_f(\alpha)=0$.

		The following theorem plays an important role in proving the subsequent 
		theorems. 

	\begin{theorem}
		(Fatou) If $f(z)\in\mathcal{A}(\mathbb{D})$ and
		$f(\mathbb{T})\subseteq\mathbb{T}$,
		then $f$ is a finite Blaschke product.
		\label{thm:fatou}
	\end{theorem}
	
	\begin{proof}
		See~\cite[Theorem 3.5.2]{garcia} and~\cite{fatou}.
	\end{proof}

	The next theorem gives a necessary and sufficient condition
	for two nonzero complex polynomials to have a constant magnitude-ratio
	on the unit circle.

	\begin{theorem}
		\label{thm:samemag}
		Let $f$ and $g\in\mathbb{C}[z]$ be two nonzero polynomials. Then $|f(z)|=\kappa|g(z)|$ for all
		$z\in\mathbb{T}$ and for some $\kappa\in\mathbb{R}^+$ if and only if
		for all $z\in\mathbb{C}\backslash\{0\}$,
		\begin{equation}
			d_f(z)+d_f(z^{-\ast})=d_g(z)+d_g(z^{-\ast}).
			\label{eq:cst_ratio}
		\end{equation}
	\end{theorem}
	
	\begin{proof}
		Suppose that $|f(z)|=\kappa|g(z)|$ 
		for a $\kappa\in\mathbb{R}^+$ and for any $z\in\mathbb{T}$.
		Let
		\begin{equation}
			B(z)=\prod_{\alpha \in Z_f\cap\mathbb{D}}B_\alpha^{d_f(\alpha)}(z)
			\label{eq:bz}
		\end{equation}
		be the Blaschke product produced by the zeros of $f$
		that are inside the unit disk. Furthermore, let
		\[H(z)\triangleq\left\{\begin{array}{lc}
			\frac{\kappa g(z)B(z)}{f(z)}, & z\in\mathbb{D};\\
			& \\
			\underset{\underset{w\in\mathbb{D}}{w\rightarrow z}}{\lim}
			\frac{\kappa g(w)B(w)}{f(w)}, & z\in\mathbb{T}.
		\end{array}\right.\]
		As $B(z)$ is a Blaschke product, it maps the unit circle to itself. Furthermore, $|f(z)|=\kappa|g(z)|$ for all $z\in\mathbb{T}$;
		consequently, $H(\mathbb{T})\subseteq\mathbb{T}$. In addition to that, the zeros of $f$ that 
		are in $\mathbb{D}$ are cancelled by $B(z)$; as a result, $H(z)\in\mathcal{A}(\mathbb{D})$. 
		By Theorem~\ref{thm:fatou}, it follows that $H(z)$ can be written
		as a finite Blaschke product, {\it i.e.}, 
				\begin{equation}
					H(z)=\tau z^{n_0}\prod_{k=1}^{p}B_{\gamma_k}^{n_k}(z),
					\label{eq:bl}
				\end{equation}
		for some finite $p\in\mathbb{N}$, some
		$\tau\in\mathbb{T}$, some distinct $\gamma_1,\ldots,\gamma_p\in\mathbb{D}\backslash\{0\}$,
		and some $n_0,\ldots,n_p\in\mathbb{N}$.
		As $\kappa g(z)B(z)=f(z)H(z)$, by substituting
		$B(z)$ from (\ref{eq:bz}) and $H(z)$ from (\ref{eq:bl}) and then multiplying
		by the denominator polynomials of $B(z)$ and $H(z)$, we have
				\begin{align}
					\kappa g(z)\prod_{\alpha\in Z_f\cap\mathbb{D}}
					(\alpha-z)^{d_f(\alpha)}\prod_{k=1}^{p}
					(1-\gamma_k^\ast z)^{n_k}=\tau f(z) z^{n_0}\prod_{\alpha\in
					Z_f\cap\mathbb{D}}(1-\alpha^\ast z)^{d_f(\alpha)}
					\prod_{k=1}^{p}(\gamma_k-z)^{n_k}.
					\label{eq:gfh3}
				\end{align}
				Let
				\[\begin{array}{ll}
					g'(z)\triangleq\underset{\alpha\in Z_f\cap\mathbb{D}}{\prod}(\alpha-z)^{d_f(\alpha)}, &
					g''(z)\triangleq\underset{k=1}{\overset{p}{\prod}}(1-\gamma_k^\ast z)^{n_k},\\
					f'(z)\triangleq\underset{\alpha\in Z_f\cap\mathbb{D}}{\prod}(1-\alpha^\ast z)^{d_f(\alpha)}, &
					f''(z)\triangleq\underset{k=1}{\overset{p}{\prod}}(\gamma_k-z)^{n_k};\\
				\end{array}\]
				then, we can write (\ref{eq:gfh3}) as 
				\begin{equation}
					\kappa g(z)g'(z)g''(z)=\tau z^{n_0}f(z)f'(z)f''(z).
					\label{eq:gfh4}
				\end{equation}
				The polynomials on both sides of (\ref{eq:gfh4}) must have the same 
				roots with the same multiplicities. 
				As a result,
				for all $z\in\mathbb{C}\backslash\{0\}$ we have
				\[d_g(z)+d_{g'}(z)+d_{g''}(z)=d_f(z)+d_{f'}(z)+d_{f''}(z),\]
				and consequently, 
				\begin{align}
					& d_g(z)+d_{g'}(z)+d_{g''}(z) + d_g(z^{-\ast})
					+ d_{g'}(z^{-\ast}) + d_{g''}(z^{-\ast}) =\nonumber\\
					& d_f(z)+d_{f'}(z)+d_{f''}(z) + d_f(z^{-\ast})
					+ d_{f'}(z^{-\ast}) + d_{f''}(z^{-\ast}).
					\label{eq:zz}
				\end{align}
				Note that $d_{g'}(z)=d_{f'}(z^{-\ast})$ and
				$d_{g''}(z)=d_{f''}(z^{-\ast})$, for all $z\in \mathbb{C}\backslash\{0\}$.
				As a result, (\ref{eq:zz}) simplifies to (\ref{eq:cst_ratio}).

				Conversely, assume that (\ref{eq:cst_ratio}) holds
				for all $z\in\mathbb{C}\backslash\{0\}$.
			If $\alpha\in Z_f\backslash\{0\}$, then $d_f(\alpha)+d_f(\alpha^{-\ast})>0$,
			which by (\ref{eq:cst_ratio}) implies that $d_g(\alpha)+d_g(\alpha^{-\ast})>0$.
			As a result, either $\alpha$ or $\alpha^{-\ast}$ belongs to $Z_g$,
			which implies that $Z_f'=Z_g'$.
			Furthermore, if $\alpha\in Z_f''$, then
			$\alpha^{-\ast}=\alpha$, which implies that $\alpha$ is a zero of $g$
			with the same multiplicity as of $f$. As a result, $Z_f''=Z_g''$.

			For some $a_f$ and $a_g\in\mathbb{C}\backslash\{0\}$
			and some $n_f$ and $n_g\in\mathbb{N}$ we have
			\begin{align*}
				f(z)=a_fz^{n_f}\prod_{\alpha\in Z_f'}(z-\alpha)^{d_f(\alpha)}
				\left(z-\alpha^{-\ast}\right)^{d_f(\alpha^{-\ast})}\prod_{\alpha\in Z_f''}(z-\alpha)^{d_f(\alpha)},
			\end{align*}
			and
			\begin{align*}
				g(z)=a_gz^{n_g}\prod_{\alpha\in Z_g'}(z-\alpha)^{d_g(\alpha)}
				\left(z-\alpha^{-\ast}\right)^{d_g(\alpha^{-\ast})}
				\prod_{\alpha\in Z_g''}(z-\alpha)^{d_g(\alpha)}.
			\end{align*}
			Let $K(z)\triangleq\frac{f(z)}{g(z)}$, then 
			\begin{align*}
				K(z)=\frac{a_f}{a_g}z^{n_f-n_g}\prod_{\alpha\in Z_f'}
				\frac{(z-\alpha)^{d_f(\alpha)-d_g(\alpha)}}{\left(z-\alpha^{-\ast}
				\right)^{d_g(\alpha^{-\ast})-d_f(\alpha^{-\ast})}}
				\prod_{\alpha\in Z_f''}
				\frac{(z-\alpha)^{d_f(\alpha)}}{(z-\alpha)^{d_g(\alpha)}},
			\end{align*}
			which, as $d_f(\alpha)=d_g(\alpha)$ for $\alpha\in Z_f''$, 
			can be simplified as 
			\begin{align}
				K(z)=\frac{a_f}{a_g}z^{n_f-n_g}\prod_{\alpha\in Z_f'}
				B_\alpha^{d_f(\alpha)-d_g(\alpha)}(z)(\alpha^\ast)^{d_f(\alpha)-d_g(\alpha)}.
				\label{eq:h2}
			\end{align}
			As a result, for all $z\in\mathbb{T}$, 
			\begin{equation}
				\left|K(z)\right|=\left|\frac{a_f}{a_g}\right|
				\prod_{\alpha\in Z_f'}
				\left|\alpha\right|^{d_f(\alpha)-d_g(\alpha)},
				\label{eq:ag}
			\end{equation}
			which is a constant number, independent of $z$. Due to 
			the definition of $K(z)$, we then have $|f(z)|=\kappa |g(z)|$,
			where $\kappa=|K(z)|$ is given in (\ref{eq:ag}).

	\end{proof}
	\begin{corollary}
		Let $f$ and $g\in\mathbb{C}[z]$ be nonzero polynomials
		such that $|f(z)|=\kappa|g(z)|$ for some
		$\kappa\in\mathbb{R}^+$ and for all $z\in\mathbb{T}$. Then
		$\deg(f)=\deg(g)$ if and only if $d_f(0)=d_g(0)$.
	\end{corollary}
	\begin{proof}
 		If $\deg(f)\neq\deg(g)$ then, according to Theorem~\ref{thm:samemag},
		the difference between the degrees can only be due to the $z$ factor. 
		The converse proof is similar.  
	\end{proof}
	In Sec.~\ref{sec:introduction}, we introduced the equivalence relations
	$\overset{\ae}{=}$ and $\overset{\phi}{\sim}$ for functions taking real
	arguments. In parallel to that, we define similar relations for functions
	that have complex arguments. 
	Two functions $f$ and $g\in\mathbb{C}^{\mathbb{C}}$ are said to be 
	{\it equal almost everywhere}, written $f\overset{\ae}{=}g$, 
	if and only if
	\[\int_{\mathbb{C}}\left|f(z)-g(z)\right|^2\d z=0.\]
	Similarly, two functions $f$ and $g\in\mathbb{C}^{\mathbb{C}}$ are 
	said to be {\it equal up to a phase offset}, written $f\overset{\phi}{\sim}g$,
	if and only if there is a $\phi\in[-\pi,\pi)$ such that $f\overset{\ae}{=}e^{i\phi}g$.
	If $f$ and $g$ are not equal up to a phase offset, we write $f\overset{\phi}{\nsim}g$.
	The relations $\overset{\ae}{=}$
	and $\overset{\phi}{\sim}$ for functions in $\mathbb{C}^{\mathbb{C}}$
	are equivalence relations. 
	
	If two polynomials $f$ and $g\in\mathbb{C}[z]$ 
	are equal a.e., then they are identical, {\it i.e.,}
	$f\overset{\ae}{=}g$ implies $f=g.$

	The next theorem plays a key role in computing a lower bound
	for the capacity of the channel that outputs $s(t;m)$ (see Fig.~\ref{fig:periodic}).

	\begin{theorem}
		\label{thm:cardinality2}
		For every $n\in\mathbb{N}$, given $f\in\mathbb{C}^{\leq n}[z]$ and $\kappa\in\mathbb{R}$, let
		$S$ be any set of complex polynomials of degree at most $n$ for which 
		$h\overset{\phi}{\nsim}g$ for all $h$ and $g\in S$, and
		$|f(z)|=\kappa|g(z)|$ for all $z\in\mathbb{T}$.
		Then
		$|S|\leq 2^{n+1}.$
	\end{theorem}
	\begin{proof}
		For a $g\in S$, let $K(z)=\frac{f(z)}{g(z)}$; then $K(z)$ can be written as in
		(\ref{eq:h2}).
		Note that by fixing $n_g$ and $d_g(\alpha)$ for all $\alpha\in Z_f'$,
		$|a_g|$ is determined uniquely by $\kappa$ from (\ref{eq:ag}).
		As $0\leq n_g\leq n-\deg(f)+d_f(0)$ and
		$0\leq d_g(\alpha)\leq d_f(\alpha)+d_f(\alpha^{-\ast})$ for all $\alpha\in Z_f'$, 
		then
		\begin{align*}
			|S|&\leq
			\left(n+1-\deg(f)+d_f(0)\right)
			\prod_{\alpha\in Z_f'}\left(d_f(\alpha)+d_f(\alpha^{-\ast})+1\right)\\
			&\leq 
			\left(n+1-\deg(f)+d_f(0)\right)
			\prod_{\alpha\in Z_f'}\left(d_f(\alpha)+d_f(\alpha^{-\ast})+1\right)
			\prod_{\alpha\in Z_f''}\left(d_f(\alpha)+1\right).
		\end{align*}
		By using the arithmetic-geometric-mean inequality we have
		\begin{align*}
			|S|\leq\left(\frac{n+|Z_f'|+|Z_f''|+1}{|Z_f'|+|Z_f''|+1}\right)^{|Z_f'|+|Z_f''|+1}
			\leq\left(\frac{n+|Z_f|+1}{|Z_f|+1}\right)^{|Z_f|+1},
		\end{align*}
		in which we have used from this property that
		\[\sum_{\alpha\in Z_f'}\left(d_f(\alpha)+
		d_f(\alpha^{-\ast})\right)+
		\sum_{\alpha\in Z_f\cap\mathbb{T}}d_f(\alpha)+d_f(0)=\deg(f).\]
		Note that $|Z_f|\leq n$ and $\left(\frac{x+\nu}{x}\right)^x$
		is an increasing function of $x$, for $x>0$; as a result,
		\[|S|\leq\left(\frac{2n+1}{n+1}\right)^{n+1}\leq 2^{n+1}.\]
	\end{proof}
\end{section}
\begin{section}{Capacity Relative to Coherent Detection}
	\label{sec:capacity_bounds}
	In this section, we find bounds for the average mutual
	information, defined in (\ref{eq:infodef}).

		For an $m\in\mathbb{N}$, let $V_m$ be the space over $\mathbb{C}$
		spanned by
		$\left\{1,\exp\left(\pm i2\pi t\right),\ldots,\exp\left(\pm i2\pi mt\right)\right\};$
		hence
		\[V_m=\left\{\sum_{k=-m}^{m}v_ke^{i2\pi kt}~:~v_k\in\mathbb{C}\right\}.\]
		Furthermore, for any $v(t)=\sum_{k=-m}^{m}v_ke^{i2\pi kt}\in V_m$, let
		\[P_v(z)\triangleq z^m\sum_{k=-m}^{m}v_kz^k.\]

	The next theorem shows that,
	up to a multiplication by a $\tau\in\mathbb{T}$, there are
	finitely many waveforms in $V_m$ that have the same magnitude 
	as $y_{p,m}(t)\in V_m$.

	\begin{theorem}
		Let $f(t)\in V_m$ be a non-zero function, and 
		let $S$ be any subset of $V_m$ such that
		$h\overset{\phi}{\nsim}g$ for all $h$ and $g\in S$,
		and $|g(t)|=|f(t)|$.
		Then 
		$|S|\leq 2^{2m+1}.$
		\label{thm:cardinality}
	\end{theorem}
	\begin{proof}
		The proof is similar to the proofs given
		in~\cite{walther,mecozzi}. Specifically,
		let $S'=\{P_g(z)~:~g(t)\in S\}$; clearly $|S'|=|S|$.
		Note that for all $g$ and $h\in S$,
		$g\overset{\phi}{\nsim}h$
		implies $P_g\overset{\phi}{\nsim}P_h$. As a result, by
		Theorem~\ref{thm:cardinality2}, $|S'|\leq 2^{2m+1}$.
	\end{proof}

	Let $Q_m:[-\pi,\pi)\rightarrow\left\{0,\pm\frac{2\pi}{m},\pm\frac{4\pi}{m},\ldots,\pm\frac{2\pi \lfloor\frac{m}{2}\rfloor}{m}\right\}$ be a
	phase-quantizer, which maps $\theta\in[-\pi,\pi)$
	to the nearest point in its range, 
	breaking ties by rotating counterclockwise. Specifically,
	\[Q_m(\theta)=\left\lfloor\frac{\theta+\frac{\pi}{m}}{\frac{2\pi}{m}}
	\right\rfloor\frac{2\pi}{m},\]
	in which $\lfloor\cdot\rfloor$ denotes the 
	floor function.
	Furthermore, for any $z\in\mathbb{C}$, let
	$\Theta_m:\mathbb{C}\rightarrow[\frac{-\pi}{m},\frac{\pi}{m})$ be defined as
	\[\Theta_m(z)\triangleq Q_m(\arg(z))-\arg(z).\]
	In another words, $\Theta_m(z)$ denotes the rotation angle which maps $z$ to the point 
	$|z|\exp(iQ_m(\arg(z)))$. 

	In Theorem~\ref{thm:cardinality}, the elements of $S$
	are not equal up to a phase offset. In order to weaken 
	this condition to have waveforms that are
	equal up to a phase offset but not everywhere,
	an auxiliary channel is introduced whose input is
	$y(t;m)$ and whose output is 
	$z(t;m)=\exp\left(i\Theta_m(b_0)\right)y(t;m).$
	As a result,
	\[z(t;m)=\sum_{k=-m}^{m}d_ke^{i2\pi kt},\]
	in which $d_k=\exp(i\Theta_m(b_0))b_k$.
	Fig.~\ref{fig:aux} shows the system, including the auxiliary channel.
		\begin{figure}
		\centering
			\includegraphics[scale=1]{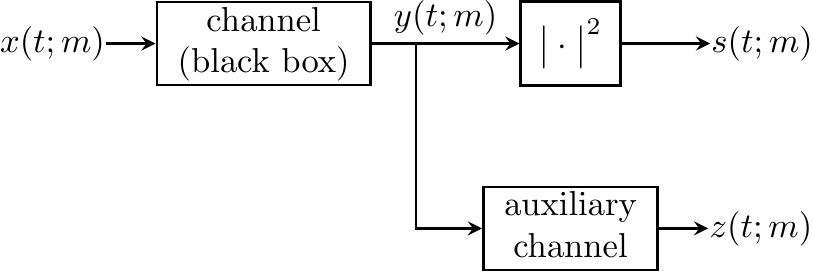}
			\caption{The relationship among different introduced waveforms.}
			\label{fig:aux}
		\end{figure}
	Let $z(t)\triangleq\lim_{m\rightarrow\infty}z(t;m)$.
	Then the system
	from $x(t)$ to $z(t)$ behaves like the coherent channel,
	{\it i.e.,} the system that observes $y(t)$. To see this, 
	for 
	$\bm{z}^{2m+1}\triangleq(d_{-m}\ldots d_m)\in\mathbb{C}^{2m+1},$
	let
	\[I(x(t;m);z(t;m))=\frac{I(\bm{x}^{2m+1};\bm{z}^{2m+1})}{2m+1},\]
	and
	\[I(x(t);z(t))=\lim_{m\rightarrow\infty}I(x(t;m);z(t;m)).\]
	Due to the chain rule for mutual information, we have
	\begin{align*}
		I(\bm{x}^{2m+1} ;\bm{y}^{2m+1},\bm{z}^{2m+1})&=
		I(\bm{x}^{2m+1} ;\bm{y}^{2m+1}) + I(\bm{x}^{2m+1};\bm{z}^{2m+1}\mid \bm{y}^{2m+1})\\
		&=I(\bm{x}^{2m+1} ;\bm{z}^{2m+1}) + I(\bm{x}^{2m+1};\bm{y}^{2m+1}\mid \bm{z}^{2m+1}).
	\end{align*}
	As $\bm{x}^{2m+1}$ --- $\bm{y}^{2m+1}$ --- $\bm{z}^{2m+1}$ is a Markov chain, we
	have $I(\bm{x}^{2m+1};\bm{z}^{2m+1}\mid \bm{y}^{2m+1})=0$, and as a result,
	\begin{align*}
		I(\bm{x}^{2m+1};\bm{y}^{2m+1}) &= 
		I(\bm{x}^{2m+1};\bm{z}^{2m+1}) + I(\bm{x}^{2m+1};\bm{y}^{2m+1}\mid \bm{z}^{2m+1})\\
		&=I(\bm{x}^{2m+1};\bm{z}^{2m+1}) + I(\bm{x}^{2m+1};\Theta_m(b_0)\mid \bm{z}^{2m+1}).
	\end{align*}
	By taking the limit as $m\rightarrow\infty$ we have	
	\begin{align*}
		I(x(t);y(t))=
		I(x(t);z(t))+
		\lim_{m\rightarrow\infty}\frac{I(\bm{x}^{2m+1};\Theta_m(b_0)\mid \bm{z}^{2m+1})}{2m+1}.
	\end{align*}
	Note that as $m\rightarrow\infty$, the interval which
	$\Theta_m(b_0)$ takes values in, {\it i.e.,} $[\frac{-\pi}{m},\frac{\pi}{m})$,
	shrinks to zero, which means that $\Theta_m(b_0)$ will take a deterministic value
	as $m\rightarrow\infty$. It implies that 
		\[\lim_{m\rightarrow\infty}\frac{I(\bm{x}^{2m+1};\Theta_m(b_0)\mid \bm{z}^{2m+1})}{2m+1}=0.\]
	
	As the channel from $x(t;m)$ to $z(t;m)$
	behaves like coherent channel when $m\rightarrow\infty$,
	instead of finding bounds for $I(\bm{x}^{2m+1};\bm{s}^{2m+1})$
	in terms of $I(\bm{x}^{2m+1};\bm{y}^{2m+1})$, we find 
	bounds in terms of $I(\bm{x}^{2m+1};\bm{z}^{2m+1})$.

		By using the chain rule for the mutual information we have
	\begin{align*}
		I(\bm{x}^{2m+1};\bm{z}^{2m+1},\bm{s}^{2m+1})&=
		I(\bm{x}^{2m+1};\bm{z}^{2m+1})+I(\bm{x}^{2m+1};
		\bm{s}^{2m+1}\mid\bm{z}^{2m+1})\\
		&=I(\bm{x}^{2m+1};\bm{s}^{2m+1})+I(\bm{x}^{2m+1};
		\bm{z}^{2m+1}\mid\bm{s}^{2m+1}).
	\end{align*}
	Note that $|z(t;m)|=|y(t;m)|$ and, as a result, 
	$s(t;m)=|z(t;m)|^2$. Consequently, 
		$\bm{x}^{2m+1}\text{ --- }\bm{z}^{2m+1}
		\text{ --- }\bm{s}^{2m+1}$ form a Markov chain.
		This implies that
		\[I(\bm{x}^{2m+1};\bm{s}^{2m+1}\mid\bm{z}^{2m+1})=0,\] 
		and as a result,
	\begin{align}
		\label{eq:equality}
		I(\bm{x}^{2m+1};\bm{s}^{2m+1})=
		I(\bm{x}^{2m+1};\bm{z}^{2m+1})-I(\bm{x}^{2m+1};
		\bm{z}^{2m+1}\mid\bm{s}^{2m+1}).
	\end{align}
		According to Theorem~\ref{thm:cardinality},
		for a particular $y_{p,m}(t)$
		and up to a constant phase ambiguity, there
		are at most $2^{2m+1}$ functions in $V_m$
		that have the same magnitude as $y_{p,m}(t)$.
		We have $y(t;m)=y_{p,m}(t)\rect(t)$,
		so up to a multiplication by some $\tau\in\mathbb{T}$,
		there are at most $2^{2m+1}$ waveforms
		of the form
		\[\left(\sum_{k=-m}^{m}g_ke^{i2\pi kt}\right)
		\rect(t),\quad g_k\in\mathbb{C},\]
		which have the same magnitude as $y(t;m)$, hence as
		$z(t;m)$.
		As a result, for the system shown in
		Fig.~\ref{fig:aux}, for a given $s(t;m)$,
		there are at most $m2^{2m+1}$
		possibilities for $z(t;m)$, where the $m$ factor multiplying $2^{2m+1}$
		is due to the $m$ possibilities for $\arg(d_0)$. Consequently,
		\begin{align*}
			I(\bm{x}^{2m+1};\bm{z}^{2m+1}\mid\bm{s}^{2m+1})
			\leq H(\bm{z}^{2m+1}\mid\bm{s}^{2m+1})
			\leq 2m+1+\log(m),
		\end{align*}
	in which $H$ denotes the entropy function.
	As a result, from (\ref{eq:equality})
	and by using the data-processing inequality, we have
	\begin{align*}
		I(\bm{x}^{2m+1};\bm{z}^{2m+1})-\left(2m+1+\log(m)\right)
		\leq I(\bm{x}^{2m+1};\bm{s}^{2m+1})
		\leq I(\bm{x}^{2m+1};\bm{y}^{2m+1}),
	\end{align*}
	thus,
	\begin{align}
		I(x(t;m);z(t;m))-1-\frac{\log(m)}{2m+1}\leq I(x(t;m);s(t;m))
		\leq I(x(t;m);z(t;m)).
		\label{eq:eq1}
	\end{align}
	By taking the limit as $m\rightarrow\infty$, (\ref{eq:eq1}) reduces to
	\begin{equation*}
		I(x(t);z(t))-1\leq I(x(t);s(t))\leq I(x(t);z(t)),
	\end{equation*}
	and as a result
	\begin{equation}
		I(x(t);y(t))-1\leq I(x(t);s(t))\leq I(x(t);y(t)).
		\label{eq:final}
	\end{equation}
	Let $p(x(t))$ denote the probability density function of 
	$x(t)$, and define 
	\[p_1\triangleq\underset{p(x(t))}{\arg\max}\quad I(x(t);y(t)),\]
	and
	\[p_2\triangleq\underset{p(x(t))}{\arg\max}\quad I(x(t);s(t)).\]
	Correspondingly, let $I_1(\cdot;\cdot)$ and
	$I_2(\cdot;\cdot)$ denote the mutual information,
	computed by $p_1$ and $p_2$, respectively. Then,
	by (\ref{eq:final}) and the definitions of $p_1$
	and $p_2$, we have
	\begin{align}
		I_1(x(t);y(t))-1\leq I_1(x(t);s(t))\leq I_2(x(t);s(t))
		\leq I_2(x(t);y(t))\leq I_1(x(t);y(t)).
		\label{eq:inequalities}
	\end{align}
	The channel capacity under coherent detection is 
	$C_\text{coh}\triangleq I_1(x(t);y(t)),$
	and under SLD it is 
	$C_\text{sld}\triangleq I_2(x(t);s(t)),$
	so we have
	\begin{equation}
		C_\text{coh}-1\leq C_\text{sld}\leq C_\text{coh}.
		\label{eq:capacity}
	\end{equation}
\end{section}

\begin{section}{Discussion}
	\label{sec:discussions}
	In Sec.~\ref{sec:setup}, we made the assumption that the
	support of $x(t)$ and $y(t)$ is limited to $[0,1)$. Restricting 
	the time interval to $[0,1)$ does not affect (\ref{eq:capacity}), 
	as in the general case, we may assume that their support 
	is $[t_1,t_2)$, for $t_1<t_2$. Then we can write $y(t)$ as 
	\[y(t)=y_{\hat{p}}(t)\rect\left(\frac{t-t_1}{t_2-t_1}\right),\]
	in which 
	\[y_{\hat{p}}=\sum_{k=-\infty}^{\infty}y(t-k(t_2-t_1))\]
	is the periodic extension of $y(t)$ with period $t_2-t_1$.
	Note that the Fourier series of $y_{\hat{p}}$ is expressed in terms of 
	$\exp\left(i\frac{2\pi k}{t_2-t_1}t\right)$,
	instead of
	$\exp\left(i2\pi kt\right)$.
	Then the
	computations are similar to the ones done for the  support $[0,1)$. 

	Although (\ref{eq:capacity}) is derived for square-law detection, the capacity bounds are true for any invertible function of $|y(t)|$,
	as well. 
	An example in which we may measure some other functions of $|y(t)|$ than $s(t)$  
	is the {\it direct detection} of optical waveform, using a photo-diode.
	Generally, diodes have a non-linear input-output relationship, in which, in certain 
	operating regimes, it might be approximated by some simple functions,
	{\it e.g.,} quadratic function. While this approximation works in those specific
	regimes, it might fail in some other. However, as long as the measurement
	is an invertible function
	of the magnitude waveform, the discussed concepts are still true.
\end{section}

\bibliography{references}
\bibliographystyle{IEEEtran}
\end{document}